\documentclass[conference,10pt,letterpaper]{IEEEtran}
\usepackage[T1]{fontenc}
\usepackage{amsmath,amssymb,amsthm,booktabs,array,cite,microtype}
\usepackage[hidelinks]{hyperref}
\usepackage{url}
\usepackage{balance}
\newcommand{\F}{\mathbb F_2}
\newcommand{\wt}{\operatorname{wt}}
\newcommand{\supp}{\operatorname{supp}}
\newtheorem{theorem}{Theorem}
\newtheorem{proposition}[theorem]{Proposition}

\newtheorem{remark}[theorem]{Remark}

\usepackage{xcolor}
\newcommand{\HX}{H_X}
\newcommand{\HZ}{H_Z}

\newcommand{\Ftwo}{\mathbb{F}_2}

\newcommand{\tr}{\mathsf{T}}
\newcommand{\vect}[1]{\mathbf{#1}}
\newcommand{\set}[1]{\mathcal{#1}}

\newcommand{\defeq}{\triangleq}

\newtheorem{lemma}{Lemma}
\newcommand{\matr}[1]{{#1}}\newcommand{\rk}{\operatorname{rank}_{\F}}
\newcommand{\vertsize}{{J}}
\newcommand{\horsize}{{L}}
\newcommand{\matrH}{\matr{H}}
\newcommand{\row}{\operatorname{row}}

\newcommand{\code}[3]

\title{Quantum LDPC Codes from Polynomial Minors}
\author{
Vahid~Nourozi$^\dag$, Roxana~Smarandache$^\ddag$,  
and David~G.~M.~Mitchell$^\dag$
\\[2ex]
$^\dag$Klipsch School of Electrical and Computer Engineering,
New Mexico State University\\
Las Cruces, NM 88003, USA,
Email: nourozi@nmsu.edu; dgmm@nmsu.edu.\\
$^\ddag$Departments of Mathematics and Electrical Engineering, University of Notre Dame\\
Notre Dame, IN 46556, USA,
Email: rsmarand@nd.edu.\vspace{-1mm}
}

\usepackage{geometry}
\begin{document}
\maketitle
\thispagestyle{empty}\pagestyle{empty}

\begin{abstract}
We construct quantum codes from codewords of a cyclic code and the minors of a polynomial matrix. The construction uses a two-row, three-column matrix whose minors yield a second matrix satisfying the CSS orthogonality condition. We derive the code dimension from the greatest common divisors of the entries and minors, with a simplified formula for odd lifting factors. A residue characterization identifies additional common factors of the minors and leads to a method for controlling the dimension through polynomial selection. We also give a condition for changing the minor generator without changing the code. Explicit polynomial constructions yield 13 quantum codes with verified minimum distances, including a $[[255,80,12]]$ code of rate $80/255$.\end{abstract} 

\section{Introduction}
The CSS construction associates a quantum code with two binary linear codes satisfying an inclusion condition~\cite{calderbank,steane}. Cyclic codes provide a convenient way to obtain such pairs. Their codewords can be represented by polynomials, and the same polynomials can be used as entries of a matrix whose binary expansion consists of circulant blocks. Moreover, minors of polynomial matrices can be used to construct generator matrices for quasi-cyclic (QC) codes as well as to determine the rank~\cite{minors, smarandache2024generalized}. 

Algebraic constructions of quantum LDPC codes include the bicycle codes of MacKay, Mitchison, and McFadden~\cite{mackay2004sparse}, the generalized bicycle (GB) codes of Kovalev and Pryadko~\cite{kovalev2013quantum}, and the constructions of Kasai et al.~\cite{kasai2011quantum} that are based on orthogonal quasi-cyclic matrices. Panteleev and Kalachev introduced lifted-product codes~\cite{panteleev2021quantum}, while Lin and Pryadko studied two-block group algebra codes over finite groups~\cite{lin2024quantum}. Bravyi et al. constructed bivariate bicycle (BB) codes from sparse bivariate polynomials 
~\cite{bravyi2024high}. Related constructions using voltage labels and structured polynomial supports appear in~\cite{nouroziapm,nouroziq,nourozilp,Nourozi2026sep4}. In particular, we will use the correspondence between polynomial and circulant transposition described in~\cite{nourozilp}. 

In this paper, we take a polynomial matrix with two rows and three columns and form its three minors. The resulting matrices must satisfy the CSS orthogonality condition for any choice of the polynomial entries using arguments from~\cite{minors,smarandache_rank,sv12}. We consider both a general polynomial matrix and a specialization in which the second row is determined by the first. The code dimension is expressed in terms of polynomial gcds. For odd $N$, it simplifies to
\begin{equation}
K=\deg g+2\deg g_1,
\label{eq:introdim}
\end{equation}
where $g$ and $g_1$ are the gcds of the entries and minors, respectively, with $x^N+1$. We characterize the additional factors of $g_1$ and show how residue conditions can be used to select the polynomials and control the dimension. We also characterize multipliers that preserve the code while changing its generators. The resulting constructions include a $[[255,80,12]]$ code, with exact distances established by explicit logical operators and exhaustive enumeration.\section{Notation and Preliminaries}
Let $N$ be the lifting factor and let
\begin{equation}
 R_N=\F[x]/(x^N+1).
\end{equation}
Polynomial sums and products are computed in $R_N$. Divisibility, factorizations, and gcds are taken in $\F[x]$, using representatives of degree less than $N$. A polynomial $a(x)=\sum_{i=0}^{N-1}a_i x^i$ also corresponds to the binary circulant matrix whose first row is $(a_0,\ldots,a_{N-1})$ and whose subsequent rows are obtained by cyclic shifts to the right.

At fixed length $N$, define
\begin{equation}
 a^\top(x)=a(x^{-1})\pmod{x^N+1}.
 \label{eq:transpose}
\end{equation}
The circulant represented by $a^\top$ is the transpose of the circulant represented by $a$. In particular, for any integer $t$
\begin{equation}
 (x^ta)^\top=x^{N-t}a^\top, 
 \label{eq:shifttranspose}
\end{equation}
where the exponent  is reduced modulo $N$. A $J\times L$  polynomial matrix  $\matrH(x) \in (R_N)^{J\times L}$
denotes  its $JN\times LN$  binary circulant expansion when its kernel, image, or binary rank is used. We write $\operatorname{Im} H$ for the binary row space of $H$, and identify binary row scalar vectors with their transposes in kernel expressions. If $H_X(x),\>H_Z(x)$ are two polynomial matrices over $R_N$,  of sizes  $J_X\times L$ and $J_Z\times L$, respectively, then, we define the classical codes
\[
\begin{aligned}
C_X&=\{u\in\mathbb F_2^n\,|\,H_Xu^T=0\},\\
C_Z&=\{u\in\mathbb F_2^n\,|\,H_Zu^T=0\},
\end{aligned}
\]
where $u$ is a row vector.   The expanded matrices have $m_X=NJ_X$, $m_Z=NJ_Z$ rows and $n=NL$ columns. Thus $J_X$ counts blocks of checks, while $m_X$ counts individual binary checks. 

The CSS condition $H_XH_Z^T=0$ is equivalent to
$C_Z^\perp\subseteq C_X$, or equivalently
$C_X^\perp\subseteq C_Z$.
When this condition holds, let
$\mathcal Q=\mathcal Q(C_X,C_Z)$ 
denote the binary
CSS code whose X-type and Z-type stabilizer
generators are specified by the rows of
$H_X$ and $H_Z$, respectively~\cite{calderbank,steane}. We refer to these stabilizers as checks. 
Let $r_X=\operatorname{rank}_{\mathbb F_2}H_X$
and $r_Z=\operatorname{rank}_{\mathbb F_2}H_Z$.
The CSS code $\mathcal Q(H_X,H_Z)$ has parameters
$[[n,k,d]]$, where $k$ is the number of encoded
qubits and $d_X,d_Z$ are its X-type and Z-type
logical distances. Explicitly,
\begin{align}
 k&=n-\rk\HX-\rk\HZ,\label{eq:cssdimension}\\
 d_X&=\min\{\wt(u)\,|\,u\in\ker\HZ\setminus\row(\HX)\},\nonumber\\
 d_Z&=\min\{\wt(u)\,|\,u\in\ker\HX\setminus\row(\HZ)\},\nonumber\\
 d&=\min(d_X,d_Z).\label{eq:distance}
\end{align}

The following theorem from~\cite{smarandache_rank, minors} gives the formula for the rank {over $\Ftwo$} of a $JN\times LN$ {scalar} matrix using the computation of the degree of the minors of the corresponding $J\times L$ polynomial matrix  $H(x)$. For subsets $\mathcal{J} \subseteq \{1, 2, \ldots, J\}$ and $\mathcal{L} \subseteq \{1, 2, \ldots, L\}$ with $|\mathcal{J}| = |\mathcal{L}|$, let $\Delta_{\mathcal{J},\mathcal{L}}(x)$ denote the determinant of the submatrix of $H(x)$ with rows indexed by $\mathcal{J}$ and columns indexed by $\mathcal{L}$. When $L = J+1$, let $\Delta_i(x)$ denote the determinant of the $J \times J$ submatrix of $H(x)$ obtained by erasing column $i$. 

\begin{theorem}[{[\citen{smarandache_rank, minors}]}]  \label{rank}  Let  $\gamma_0\defeq 1$, and, for all $1\leq i\leq J$, let  the following be defined  in $\Ftwo[x]$: 
\begin{align*}&\gamma_i(x)\defeq \gcd\{ \Delta_{\set{J},\set{L}}(x)\mid  |\set{J}|=|\set{L}|=i\},\\ 
&d_i(x) \defeq \gcd (\gamma_i(x)/\gamma_{i-1}(x), x^N+1), \text{ where }\\ 
&\gcd(0/0,x^N+1)\defeq \gcd(0,x^N+1)=x^N+1.\end{align*} Then,\begin{align*}
&\rk(H)= 
J\cdot N -\sum_{i=1}^{J} \deg d_i(x).\\ 
\end{align*}\end{theorem}

The following Lemma  from \cite{sv12, smarandache_rank} gives a way to create codewords starting from a given matrix $H(x)$ and using minors. 
\begin{lemma}
  \label{lemma:codeword:construction:1}
Let $S \subseteq \{1, 2, \ldots, L\}$ be any subset of size $J+1$. Define the length $L$ vector $\vect{c}_S(x) = \left(c_{S,1}(x), c_{S,2}(x), \ldots, c_{S,L}(x)\right)$ over $R_N$ by $c_{S,i}(x) \triangleq \det H_{S \setminus \{i\}}(x)$ if $i \in S$, and $0$ otherwise, where $H_{S \setminus \{i\}}(x)$ is the $J \times J$ submatrix of $H(x)$ with columns indexed by $S \setminus \{i\}$. Then $\vect{c}_S(x)$ (as a column vector) lies in $\ker H(x)$.

\end{lemma}

 \begin{theorem} \label{minors-code} Let $H_X (x)\in (R_N)^{J_X\times L}$. For a fixed  $\set{S}\subset \{1, 2, \ldots, \horsize\}$
of size $\vertsize_X $, form the column vectors  $\vect{c}_{\set S\cup \{i\}}(x)$, for all $i\notin \set S$. Define $H_Z(x) \in (R_N)^{J_Z\times L}$  with rows given by the vectors  
 $\vect{c}^\top_{\set S\cup \{i\}}(x)$.  Then $H_XH_Z^\top=0$ and $C_X^\perp\subseteq C_Z$, so the length $NL$ code $\mathcal Q (C_X, C_Z)$ exists. If \[
\operatorname{rank}H_X+\operatorname{rank}H_Z<LN, 
\] then its dimension is positive. 
 
 In particular, if $L=J_X+1$, $H_Z=\begin{bmatrix} \Delta^\tr_1& \Delta^\tr_2&\ldots&\Delta^\tr_L\end{bmatrix}$. If $L = J_X + 1$ and  $\gamma_{i,X}$ and  $d_{i,X}$ are defined as in Theorem~\ref{rank} for the matrix $H_X$, for all $1\leq i\leq J_X$,  then the dimension of  the  quantum code is
 \begin{equation}\label{dimension} K=\sum_{i=1}^{J} \deg d_{i,X}(x) + \deg  \gcd(\gamma_{J_X,X}, x^N+1).\end{equation}
 \end{theorem}  
  \begin{proof}  $\gamma_{J_X,X}$ is the gcd of all $J_X\times J_X$ minors $\Delta_i$ of $H_X$, and, therefore, it is the gcd of all the entries of the matrix $H_Z^T$, giving $d_{1,Z}(x)= \gcd(\gamma_{J_X,X}, x^N+1)$. 
  Then $ \rk(H_X)=  
J_X N -\sum_{i=1}^{J_X} \deg d_{i,X}(x) $ and $\rk(H_Z)= \rk(H_Z^T)= N - \deg  \gcd(\gamma_{J_X,X}, x^N+1)$.
The dimension of $\mathcal Q$ now follows from \eqref{eq:cssdimension}. 
\end{proof}
 \begin{remark}\label{remark3}
 Since $C_Z^\perp$ is also a subset of a supercode $C'_X\defeq \ker (H'_X)$ of $C_X$, where $H'_X$ is a $J'\times L$ polynomial submatrix of $H_X$ obtained by taking a smaller set of rows, then the code $\mathcal  Q'(C'_X, C_Z) $ can be constructed with a larger dimension, and possibly the same minimum distance, as the next example shows. 

The second remark is that the generalized bicycle code \cite{kovalev2013quantum} given by $$H_X=\begin{bmatrix} a(x)& b(x) \end{bmatrix}, H_Z=\begin{bmatrix} b^\tr(x)& a^\tr(x) \end{bmatrix}$$ is an example of a code obtained by Theorem~\ref{minors-code}. 
  \end{remark}

  Before presenting our main construction, we exemplify the results in the following example.
  
  \emph{Example:} We take $N$ variable, and
\begin{align*}&H(x)=  \begin{bmatrix}   1+x^2 &1+x^4& 1+x^6& 1+x^8\\
                            x^4+x^{12}& x^{20} +x^{22}& x^{30}+ x^{42}& x^{40} +x^{14}\\
                             1+x^4 &x^{30}+x^{24}& x^{12}+x^{14}& x^3+x^{13} \end{bmatrix}. \end{align*}     
We compute
\begin{align*}
&\gamma_1=x^2+1, \quad \gamma_2=(x^2+1)^2, \quad \gamma_3=(x^2+1)^3, \\ 
&\frac{\gamma_1}{\gamma_0}=\frac{\gamma_2}{\gamma_1}=\frac{\gamma_3}{\gamma_2}= x^2+1 \Longrightarrow \text{ for all } 1\leq i\leq 3,  \\
&d_i=\gcd(x^2+1, x^N+1)=\left\{\begin{matrix} x^2+1, & \text{ if }  N \text{ even} \\
 x+1,& \text{ if }  N \text{ odd}
\end{matrix}\right.. \end{align*}

      Let \begin{align*} H_Z &\defeq \begin{bmatrix}\Delta^\tr_{1}& \Delta^\tr_2& \Delta^\tr_3& \Delta^\tr_4\end{bmatrix},\\   
  C_X&=\ker H_X, \quad   C_Z= \ker H_Z,\\
\mathcal Q &=\mathcal Q(C_X, C_Z).\end{align*}       
For $N=44$, the gcd of all the entries in $H_Z$ and $x^{N}+1$  is $(x^2+1)^2$. 
Therefore, \eqref{dimension} gives  $K=6+4=10$. Indeed, this code has parameters $[[176,10,8]]$. 
    Enlarging  $C_X$ to be equal to the kernel of only the first two rows of $H_X$, since $C_Z$ is still a subcode, we can form the quantum code of larger dimension $[[176, 52, 8]]$. $\hfill \Box$\vspace{3mm}
  
In the remainder of the paper, starting from a $2\times 3$ polynomial matrix,  we will use this idea in order  to construct a class of quantum codes of length $3N$.

\section{Construction from Polynomial Minors}
Let $g(x)\mid x^N+1$ and let $a,b,c,d,e,f$ be codewords of the cyclic code $\langle g(x) \rangle$ generated by $g(x)$. Take $$ C_X=\ker H_X, \quad C_Z=\ker H_Z,$$
\begin{equation}
 H_X=\begin{bmatrix}
 a^\tr&b^\tr&c^\tr\\
 d^\tr&e^\tr&f^\tr \end{bmatrix},
\end{equation}   
\begin{equation}
 H_Z= \begin{bmatrix}\Delta_1& \Delta_2 &\Delta_3\end{bmatrix}=\begin{bmatrix}bf+ce&af+cd&ae+bd\end{bmatrix}.
 \label{eq:generalC2}
\end{equation}
The following is an immediate consequence of Theorem~\ref{minors-code}. 
\begin{proposition}\label{prop:orthogonality}
The codes above satisfy $C_Z^\perp\subseteq C_X$. Hence $\mathcal Q = \mathcal Q 
(C_X,C_Z)$ is a quantum code of length $3N$.
\end{proposition}

Taking the entries to be codewords in  a cyclic code $\langle g(x)\rangle$ generated by $ g(x)$ gives a way to choose them with a prescribed common factor to increase the dimension. 
In particular, to minimize the number of codewords in $\langle g(x) \rangle$ that we need to choose, we take $d=x^tc$, 
$e=a$, and $f=b$. Then

\begin{align}
 H_X&=\begin{bmatrix}
 a^\tr&b^\tr&c^\tr\\
 x^{-t}c^\tr&a^\tr&b^\tr
 \end{bmatrix},\label{eq:HXshift}\\
 H_Z&=\begin{bmatrix}b^2+ac&ab+x^tc^2&a^2+x^tbc\end{bmatrix}.
 \label{eq:HZshift}
\end{align}
Thus, $t=N-1$ gives the entry $xc^\tr$, while $t=0$ gives the entry $c^\tr$.

\subsection{Weights of the Matrices}
Let $w_a=\wt(a)$, $w_b=\wt(b)$, and $w_c=\wt(c)$. Each row of $H_X$ has weight $w_a+w_b+w_c$, and the column weights in its three blocks are $w_a+w_c$, $w_a+w_b$, and $w_b+w_c$. The columns in block $i$ of $H_Z$ have weight $\wt(\Delta_i)$, and each row has weight
\[
\sum_{i=1}^3\wt(\Delta_i)
\le w_a+w_b+w_c+w_aw_b+w_aw_c+w_bw_c.
\]
\subsection{Dimension and Common Factors}
Define
$$g=\gcd(a,b,c,x^N+1)$$
and assume that at least one minor $\Delta_i$ is nonzero. Since $g$ divides $a,b,c$, it follows that $g^2$ divides each minor. Hence there is a nonzero polynomial $F\in\Ftwo[x]$ such that
$$g^2F=\gcd(\Delta_1,\Delta_2,\Delta_3).$$

\begin{theorem}\label{thm:dimension}
Let \begin{equation} \label{eq:gcds} g_1(x)=\gcd(\Delta_1(x),\Delta_2(x),\Delta_3(x),x^N+1)\end{equation} and  \[
h=\gcd(gF,x^N+1).
\]
The dimension of $\mathcal Q$ is
\[
K=\deg g+\deg h+\deg g_1.
\]
In particular, for odd $N$,
\begin{equation}\label{eq:factorK}K=\deg g+2\deg g_1
=3\deg g+2\deg f,
\end{equation}
where $f=g_1/g$.
\end{theorem}

\begin{proof}
The Smith invariants of $H_X$ are $g$ and
$gF$, while the nonzero Smith invariant
of $H_Z$ is $g^2F$. 
Applying~\eqref{dimension} gives $K=\deg g+\deg h +\deg g_1$. 
For odd $N$, the polynomial $x^N+1$
is square-free. Since $g\mid x^N+1$,
\[
h=\gcd(gF,x^N+1)
=\gcd(g^2F,x^N+1)=g_1.
\]
The result follows.
\end{proof}

\subsection{Controlling the Dimension}
For odd $N$ and fixed $g$, the dimension increases with $\deg f$.
For $t=0$, one way to obtain an additional factor is to choose
\begin{equation}
 a(x)\equiv b(x)\equiv c(x)\equiv s(x)\not\equiv0\pmod{p(x)}
 \label{eq:equal}
\end{equation}
for an irreducible 
$p(x)\mid x^N+1$. The two polynomial rows then coincide modulo $p(x)$, and all the minors vanish. The following proposition gives all nonzero possibilities.

\begin{proposition}\label{prop:residues}
Let $p\mid x^N+1$ be irreducible, and suppose that not all of $a,b,c$ vanish modulo $p$. The three minors of~\eqref{eq:HXshift} vanish modulo $p$ if and only if there are nonzero $s,\lambda\in\F[x]/(p)$ such that
\begin{equation}
 (a,b,c)=s(1,\lambda^{-1},\lambda^{-2}),\qquad \lambda^3=x^t
 \label{eq:cubic}
\end{equation}
in that field. In particular, for $t=0$ this becomes
\begin{equation}
 (a,b,c)=s(1,\lambda^2,\lambda),\qquad \lambda^3=1.
 \label{eq:unity}
\end{equation}
\end{proposition}
\begin{proof}
Both rows of~\eqref{eq:HXshift} are nonzero. Their minors vanish in the field $\F[x]/(p)$ precisely when the second row is $\lambda$ times the first. Thus
\begin{equation}
 x^tc=\lambda a,\qquad a=\lambda b,\qquad b=\lambda c.
\end{equation}
All three entries and $\lambda$ are nonzero in $\F[x]/(p)$. Setting $s=a$ yields~\eqref{eq:cubic}. Conversely,~\eqref{eq:cubic} makes the rows proportional.
\end{proof}
For $t=0$, the choice $\lambda=1$ yields~\eqref{eq:equal}. If $\deg p$ is even, there are two additional cube roots of unity, giving unequal residues. For example, with $p=1+x+x^2$, choosing $s=1$ and $\lambda=x$ yields
\[
a\equiv1,\qquad b\equiv x^2,\qquad c\equiv x\pmod p.
\]
If $\deg p$ is odd, $\lambda=1$ is the only solution.
For general $t$, the equation $\lambda^3=x^t$ has a unique solution when $\deg p$ is odd. When $\deg p=r$ is even, it has three solutions if
\[
(x^t)^{(2^r-1)/3}\equiv1\pmod p,
\]
and none otherwise. Thus the choice of $t$ determines which even-degree irreducible factors can contribute to $f$.
\subsection{Choosing the Polynomials}
For odd $N$, the dimension can be controlled by selecting
coprime factors $g,f$ of $x^N+1$. For each irreducible
factor $p$ of $g$, choose $a,b,c$ to be zero modulo $p$.
For each factor $p$ of $f$, choose nonzero residues
satisfying Proposition~\ref{prop:residues}. The reduced
matrix in~\eqref{eq:HXshift} then has rank zero modulo
the factors of $g$ and rank one modulo the factors of $f$.
To obtain $g_1=gf$ exactly, require rank two over
$\F[x]/(p)$ for every remaining irreducible factor $p$
of $x^N+1$. The Chinese remainder theorem combines
these local conditions into polynomials $a,b,c$.

The resulting polynomials need not have low weight.
To obtain sparse entries, search among codewords in
$\langle g\rangle$ whose residues satisfy the conditions
for the factors of $f$, without imposing rank two
at the remaining factors. Additional common factors
may then occur, so the gcds must be recomputed.
The dimension follows from~\eqref{eq:factorK};
the minimum distance is determined separately
from~\eqref{eq:distance}.

\subsection{Replacing the Generator Matrix}
Let $z=(\Delta_1,\Delta_2,\Delta_3)$. For $q\in R_N$, replacing $H_Z$ by the circulant expansion of $qz$ preserves orthogonality. The following condition characterizes when it also preserves $C_Z^\perp$.\pagebreak
\begin{proposition}\label{prop}
The row spaces of the circulant expansions of $z$ and $qz$ are equal if and only if
\begin{equation}
\gcd\left(q,\frac{x^N+1}{g_1}\right)=1.
\label{eq:q}
\end{equation}
\end{proposition}
\begin{proof}
The row space of the circulant expansion of $z$ is $R_N z = \{rz \mid r \in R_N\}$, and $R_N qz \subseteq R_N z$. The map $r \mapsto rz$ from $R_N$ onto $R_N z$ has kernel $\{r \mid r\Delta_i \equiv 0 \pmod{x^N+1},\ i = 1,2,3\}$. For each $i$, $r\Delta_i \equiv 0$ if and only if $(x^N+1)/\gcd(\Delta_i, x^N+1)$ divides $r$, so the kernel is generated by
\[
\operatorname{lcm}_i \frac{x^N+1}{\gcd(\Delta_i, x^N+1)} = \frac{x^N+1}{g_1}.
\]
Hence $R_N z \cong \mathbb{F}_2[x]/\left((x^N+1)/g_1\right)$, with $z$ corresponding to $1$ and $qz$ to $q$. Thus $R_N qz = R_N z$ if and only if $q$ is a unit in this quotient, which is equivalent to \eqref{eq:q}.
\end{proof}

An admissible $q$ preserves the two classical codes and hence $K,d_X,d_Z$, but may change the generator weights. All multipliers used in the examples satisfy~\eqref{eq:q}.

\section{Constructed Codes}The following constructions illustrate how common factors of the minors can be used to control the dimension while the polynomial entries are chosen with low weight. Tables~\ref{tab:parameters} and~\ref{tab:polynomials} give the parameters and polynomials of the constructed codes. Every row specifies the construction in~\eqref{eq:HXshift} and~\eqref{eq:HZshift}, with an optional multiplier $q$ as in Proposition~\ref{prop}. For an exponent set $S$, the corresponding polynomial is $\sum_{i\in S}x^i$. All exponents start at zero and are reduced modulo $N$. The rows of $H_X$ have low weight, while the rows of $H_Z$ are products of entries and are heavier. Table~\ref{tab:parameters} lists both.

\subsection{A Code with Parameters \texorpdfstring{$[[255,80,12]]$}{[[255,80,12]]}}
Let $N=85$ and take
\begin{align}
 a(x)&=1+x^{17}+x^{24}+x^{34}+x^{40}+x^{44}+x^{45},\nonumber\\
 b(x)&=x^5+x^{22}+x^{29}+x^{39}+x^{45}+x^{49}+x^{50},\nonumber\\
 c(x)&=x+x^{38}+x^{41}+x^{46}+x^{55}+x^{65}+x^{74}+x^{84}.
 \label{eq:255polynomials}
\end{align}
Use these polynomials in~\eqref{eq:HXshift} and~\eqref{eq:HZshift} with $t=0$ and $q=1$. Their gcd is
\begin{align}
 g(x)={}&1+x^3+x^4+x^5+x^6+x^{10}+x^{14}+x^{15}\nonumber\\
        &+x^{16}+x^{17}+x^{21}+x^{22}+x^{24}.
\end{align}
The common gcd of the minors with $x^{85}+1$ is $g_1=gf$, where
$$f(x)=1+x+x^2+x^3+x^4.$$
Moreover,
\begin{equation}
b=x^5a,\qquad a\equiv b\equiv c\equiv x+x^3\pmod f.
\label{eq:255residues}
\end{equation}
Since $f\mid x^5+1$, these residues satisfy the rank-one condition.

It follows that $\deg g=24,\deg f=4$, and
\begin{equation}
K=3(24)+2(4)=80, \quad \frac{K}{3N}=\frac{80}{255}\simeq0.3137.\end{equation}
The binary ranks are $\rk H_X=118$ and $\rk H_Z=57$. 
Every row of $H_X$ has weight $22$, and every row of $H_Z$ has weight $61$.

An explicit nontrivial logical operator is represented by
\begin{equation}
 \ell_X=(\ell_1,0,\ell_3),
\end{equation}
where
\begin{align}
 \ell_1&=1+x^{14}+x^{39}+x^{40}+x^{45}+x^{51}+x^{54}+x^{62},\nonumber\\
 \ell_3&=x^{16}+x^{25}+x^{29}+x^{60}.
\end{align}
This vector has weight $12$ and lies in $\ker H_Z$. Appending it to $H_X$ increases the rank from $118$ to $119$, so it is not in $\operatorname{Im} H_X$.

The all-ones vector lies in $\operatorname{Im}H_Z$, so every word in $\ker H_Z$ has even weight. Exhaustive column-sum enumeration excludes weights two through eight. Any weight ten word has at least four positions in one block, including a pair with cyclic gap at most $21$. Up to cyclic shift, there are $63$ such pairs. Matching the remaining supports of size four excludes weight ten.

\begin{table}[t]
\caption{Quantum codes obtained from the construction. Parameters $w_X$ and $w_Z$ are the row weights of $H_X$ and $H_Z$. All listed minimum distances are exact.}
\label{tab:parameters}
\centering\small\renewcommand{\arraystretch}{1.06}\setlength{\tabcolsep}{4.5pt}
\begin{tabular}{c c c c c c c}
\toprule
Example & $[[n,K,d]]$ & $K/n$ & $\deg g$ & $\deg g_1$ & $w_X$ & $w_Z$\\
\midrule
1 & $[[27,4,4]]$ & 0.1481 & 0 & 2 & 8 & 11\\
2 & $[[45,16,4]]$ & 0.3556 & 4 & 6 & 10 & 19\\
3 & $[[63,15,6]]$ & 0.2381 & 5 & 5 & 10 & 21\\
4 & $[[93,15,8]]$ & 0.1613 & 5 & 5 & 10 & 25\\
5 & $[[93,17,7]]$ & 0.1828 & 5 & 6 & 9 & 22\\
6 & $[[135,14,9]]$ & 0.1037 & 4 & 5 & 9 & 24\\
7 & $[[135,22,8]]$ & 0.1630 & 4 & 9 & 9 & 30\\
8 & $[[135,14,9]]$ & 0.1037 & 4 & 5 & 9 & 22\\
9 & $[[135,14,10]]$ & 0.1037 & 4 & 5 & 9 & 32\\
10 & $[[135,32,9]]$ & 0.2370 & 10 & 11 & 35 & 38\\
11 & $[[153,51,8]]$ & 0.3333 & 17 & 17 & 18 & 38\\
12 & $[[189,58,10]]$ & 0.3069 & 18 & 20 & 16 & 63\\
13 & $[[255,80,12]]$ & 0.3137 & 24 & 28 & 22 & 61\\

\bottomrule
\end{tabular}

\end{table}

\begin{table*}[t]
\caption{Polynomial entries. The entries in each set are the exponents with nonzero coefficients.}
\label{tab:polynomials}
\centering\small\renewcommand{\arraystretch}{1.04}\setlength{\tabcolsep}{3pt}
\begin{tabular}{c r r >{\raggedright\arraybackslash}p{1.70in} >{\raggedright\arraybackslash}p{1.70in} >{\raggedright\arraybackslash}p{1.70in} c}
\toprule
Example & $N$ & $t$ & $\supp(a)$ & $\supp(b)$ & $\supp(c)$ & $q$\\
\midrule
1 & 9 & 0 & $\{0,\allowbreak 1,\allowbreak 3\}$ & $\{0,\allowbreak 1\}$ & $\{2,\allowbreak 5,\allowbreak 6\}$ & $1$ \\
2 & 15 & 0 & $\{0,\allowbreak 3,\allowbreak 14\}$ & $\{6,\allowbreak 7,\allowbreak 11,\allowbreak 14\}$ & $\{4,\allowbreak 11,\allowbreak 13\}$ & $1$ \\
3 & 21 & 0 & $\{2,\allowbreak 4,\allowbreak 11,\allowbreak 16\}$ & $\{13,\allowbreak 17,\allowbreak 18\}$ & $\{5,\allowbreak 9,\allowbreak 10\}$ & $q_{3}$ \\
4 & 31 & 0 & $\{6,\allowbreak 15,\allowbreak 24,\allowbreak 27\}$ & $\{12,\allowbreak 21,\allowbreak 28\}$ & $\{4,\allowbreak 17,\allowbreak 18\}$ & $q_{4}$ \\
5 & 31 & 0 & $\{8,\allowbreak 15,\allowbreak 30\}$ & $\{19,\allowbreak 21,\allowbreak 24\}$ & $\{25,\allowbreak 27,\allowbreak 30\}$ & $1$ \\
6 & 45 & 0 & $\{3,\allowbreak 5,\allowbreak 26\}$ & $\{0,\allowbreak 7,\allowbreak 9\}$ & $\{17,\allowbreak 39,\allowbreak 41\}$ & $1$ \\
7 & 45 & 0 & $\{27,\allowbreak 28,\allowbreak 31\}$ & $\{22,\allowbreak 30,\allowbreak 39\}$ & $\{12,\allowbreak 22,\allowbreak 32\}$ & $1$ \\
8 & 45 & 42 & $\{14,\allowbreak 15,\allowbreak 33\}$ & $\{0,\allowbreak 2,\allowbreak 23\}$ & $\{7,\allowbreak 14,\allowbreak 31\}$ & $1$ \\
9 & 45 & 1 & $\{20,\allowbreak 29,\allowbreak 42\}$ & $\{0,\allowbreak 33,\allowbreak 44\}$ & $\{0,\allowbreak 6,\allowbreak 43\}$ & $1$ \\
10 & 45 & 0 & $\{3,\allowbreak 4,\allowbreak 6,\allowbreak 10,\allowbreak 12,\allowbreak 13,\allowbreak 15,\allowbreak 17,\allowbreak 18,\allowbreak 19,\allowbreak 20,\allowbreak 23,\allowbreak 27\}$ & $\{4,\allowbreak 5,\allowbreak 8,\allowbreak 13,\allowbreak 14,\allowbreak 17,\allowbreak 29,\allowbreak 30,\allowbreak 32,\allowbreak 35,\allowbreak 39\}$ & $\{9,\allowbreak 10,\allowbreak 12,\allowbreak 13,\allowbreak 14,\allowbreak 15,\allowbreak 16,\allowbreak 20,\allowbreak 21,\allowbreak 26,\allowbreak 30\}$ & $q_{10}$ \\
11 & 51 & 0 & $\{3,\allowbreak 7,\allowbreak 20,\allowbreak 24,\allowbreak 37,\allowbreak 41\}$ & $\{10,\allowbreak 17,\allowbreak 20,\allowbreak 27,\allowbreak 38,\allowbreak 50\}$ & $\{16,\allowbreak 18,\allowbreak 27,\allowbreak 30,\allowbreak 42,\allowbreak 50\}$ & $q_{11}$ \\
12 & 63 & 0 & $\{5,\allowbreak 17,\allowbreak 30,\allowbreak 34,\allowbreak 35,\allowbreak 42\}$ & $\{15,\allowbreak 17,\allowbreak 23,\allowbreak 47,\allowbreak 57\}$ & $\{4,\allowbreak 9,\allowbreak 51,\allowbreak 52,\allowbreak 55\}$ & $1$ \\
13 & 85 & 0 & $\{0,\allowbreak 17,\allowbreak 24,\allowbreak 34,\allowbreak 40,\allowbreak 44,\allowbreak 45\}$ & $\{5,\allowbreak 22,\allowbreak 29,\allowbreak 39,\allowbreak 45,\allowbreak 49,\allowbreak 50\}$ & $\{1,\allowbreak 38,\allowbreak 41,\allowbreak 46,\allowbreak 55,\allowbreak 65,\allowbreak 74,\allowbreak 84\}$ & $1$ \\

\bottomrule 
\end{tabular}
\par\vspace{5pt}
\noindent\begin{minipage}{0.97\textwidth}\footnotesize
The support of each nonconstant multiplier is given below. Each satisfies~\eqref{eq:q}.
\begin{align*}
\supp(q_{3})&=\{0,7,8\},\\
\supp(q_{4})&=\{0,2,4,5,8,12,13,19,20\},\\
\supp(q_{10})&=\{4,5,6,10,11,12,16,18,19,20,21,25,26,27,29,32\},\\
\supp(q_{11})&=\{3,7,8,11,14,16,18,20,21,22,27,29\}.
\end{align*}
\end{minipage}\vspace{-5mm}
\end{table*}

For $d_Z$, let $D = \langle (x^{85}+1)/g \rangle$, so that $d(D) = 24$. The words of $\ker H_X$ supported on a single block form $D$ for the first and third blocks, and a code of dimension $32$ and minimum distance $16$ for the second block, since $b = x^5 a$. A nonzero word of weight at most $11$ has a block of weight at most three. Fixing this block up to cyclic shift, the other two blocks lie in one of $16$ cosets of the product of their single block codes, and the minimum weight of each coset splits over the two blocks. Checking all such cosets, together with the words that have a zero block, shows that $\ker H_X$ has no nonzero word of weight at most $11$. Hence $d_Z \geq 12$. 
A $Z$ logical operator of weight $20$ gives $d_Z\le20$. Hence
$$d_X=12,\qquad 12\le d_Z\le20,$$
and the code has parameters $[[255,80,12]]$.

\subsection{Further Choices of the Common Factors} For Example 12, let $N=63$ and take
\begin{align}
 a(x)&=x^5+x^{17}+x^{30}+x^{34}+x^{35}+x^{42},\nonumber\\
 b(x)&=x^{15}+x^{17}+x^{23}+x^{47}+x^{57},\nonumber\\
 c(x)&=x^4+x^9+x^{51}+x^{52}+x^{55}.
 \label{eq:189polynomials}
\end{align}
The gcd of the minors with $x^{63}+1$ is $g_1=gf$, where
\begin{align}
 g(x)={}&1+x^3+x^6+x^7+x^8+x^{10}\nonumber\\
     &+x^{13}+x^{14}+x^{15}+x^{16}+x^{18},\nonumber\\
 f(x)={}&1+x+x^2.
\end{align}
Unlike Example 13, Example 12 uses distinct nonzero residues modulo $f$:
\begin{equation}
 a\equiv1,\qquad b\equiv x^2,\qquad c\equiv x\pmod f.
\end{equation}
Since $x^3=1\pmod f$, the residues satisfy~\eqref{eq:unity} with $\lambda=x$. Hence
\begin{align}
 &K=3(18)+2(2)=58,\nonumber\\ &K/(3N)=58/189\simeq0.3069.
\end{align}
The binary ranks are $88$ and $43$. The row weights of $H_X$ and $H_Z$ are 16 and 63, respectively, with \(q=1\).

The support
\begin{equation}
 \{16,78,93,97,110,115,131,136,166,185\}
\end{equation}
defines a weight $10$ vector in $\ker H_Z\setminus\operatorname{Im} H_X$ (coordinates start at zero). Exhaustive support enumeration excludes nontrivial $X$ operators of weight at most eight. The all-ones vector lies in the row space of $H_Z$, excluding odd weights. For $d_Z$, the words of $\ker H_X$ supported on a single block form $\langle (x^{63}+1)/g \rangle$, of minimum distance $16$, for the first and second blocks, and a code of dimension $21$ and minimum distance $12$ for the third block. Fixing a block of weight at most three, as in Example~13, shows that $\ker H_X$ has no nonzero word of weight at most $11$. A logical operator of weight $12$ supported on the third block, for example the support $\{126, 132, 135, 138, 147, 153, 156, 159, 168, 174, 177, 180\}$, attains the bound. Thus $d_X=10$, $d_Z=12$, giving $[[189,58,10]]$.

Example 11 has $N=51$ and $g=g_1$ of degree $17$. Thus $K=3(17)=51$ and the rate is $1/3$. The three polynomials in Table~\ref{tab:polynomials} each have weight six. The multiplier $q_{11}$ gives a generator of the same $C_Z^\perp$ with row weight $38$. The exact $X$ and $Z$ distances are $d_X=8$ and $d_Z=12$, giving $[[153,51,8]]$.

Finally, for Example 10, $N=45$, $\deg g=10$, and $\deg g_1=11$, so $K=32$. The multiplier $q_{10}$ preserves $C_Z^\perp$ and gives row weight $38$. The exact distances $d_X=9$ and $d_Z=10$ give $[[135,32,9]]$. For both examples, exhaustive support enumeration and affine-class calculations establish the $X$ and $Z$ lower bounds, respectively. Explicit operators attain them.

\subsection{Choosing Codewords of Small Weight}

To find low-weight multiples of $g$, group small sets of exponents by their remainders modulo $g$. For an exponent set $S$, let
$$r(S)=\sum_{i\in S}x^i\pmod g.$$
If $r(S)=r(T)$, then
\begin{equation}
g\mid\sum_{i\in S}x^i+\sum_{j\in T}x^j.
\label{eq:smallcodewords}
\end{equation}
Repeated exponents cancel. Sets of sizes three and three yield a codeword of weight at most six, while sizes two and three yield weight at most five. This method produces the entries in Examples 11 and 12. For Example 12, residues $(1,x^2,x)$ are selected modulo $1+x+x^2$, followed by verification of the gcds and distances.

\section{Conclusion}
Polynomial minors of a $2 \times 3$ matrix give an explicit inclusion $C_Z^{\perp} \subseteq C_X$ and hence a CSS quantum code of length $3N$. For odd $N$, the common factors of the entries and of the minors determine the dimension, and residue conditions on the polynomials can be used to reach a target dimension. The resulting codes include $[[255, 80, 12]]$ and $[[189, 58, 10]]$, with verified exact minimum distances. The rows of $H_Z$ are heavier than those of $H_X$. Multipliers that preserve the code can reduce this weight, and constructions with low weight in both matrices, and their decoding performance, are left for future work.
\balance
\section*{Acknowledgement}
This material is based upon work supported by the National Science Foundation under Grant No. CCF-2145917. 



\begin{thebibliography}{10}
\providecommand{\url}[1]{#1}
\csname url@samestyle\endcsname
\providecommand{\newblock}{\relax}
\providecommand{\bibinfo}[2]{#2}
\providecommand{\BIBentrySTDinterwordspacing}{\spaceskip=0pt\relax}
\providecommand{\BIBentryALTinterwordstretchfactor}{4}
\providecommand{\BIBentryALTinterwordspacing}{\spaceskip=\fontdimen2\font plus
\BIBentryALTinterwordstretchfactor\fontdimen3\font minus \fontdimen4\font\relax}
\providecommand{\BIBforeignlanguage}[2]{{%
\expandafter\ifx\csname l@#1\endcsname\relax
\typeout{** WARNING: IEEEtran.bst: No hyphenation pattern has been}%
\typeout{** loaded for the language `#1'. Using the pattern for}%
\typeout{** the default language instead.}%
\else
\language=\csname l@#1\endcsname
\fi
#2}}
\providecommand{\BIBdecl}{\relax}
\BIBdecl

\bibitem{calderbank}
A.~R. Calderbank and P.~W. Shor, ``Good quantum error-correcting codes exist,'' \emph{Physical Review A}, vol.~54, no.~2, p. 1098, 1996.

\bibitem{steane}
A.~Steane, ``Multiple-particle interference and quantum error correction,'' \emph{Proceedings: Mathematical, Physical and Engineering Sciences}, pp. 2551--2577, 1996.

\bibitem{minors}
R.~Smarandache, A.~G{\'o}mez-Fonseca, and D.~G. Mitchell, ``Using minors to construct generator matrices for quasi-cyclic {LDPC} codes,'' \emph{Proc. IEEE International Symposium on Information Theory (ISIT)}, pp. 548--553, 2022.

\bibitem{smarandache2024generalized}
R.~Smarandache, A.~G{\'o}mez-Fonseca, and D.~G. Mitchell, ``Generalized quasi-cyclic {LDPC} codes: Design and efficient encoding,'' \emph{Proc. IEEE International Symposium on Information Theory (ISIT)}, pp. 434--439, 2024.

\bibitem{mackay2004sparse}
D.~J. MacKay, G.~Mitchison, and P.~L. McFadden, ``Sparse-graph codes for quantum error correction,'' \emph{IEEE Transactions on Information Theory}, vol.~50, no.~10, pp. 2315--2330, 2004.

\bibitem{kovalev2013quantum}
A.~A. Kovalev and L.~P. Pryadko, ``Quantum {Kronecker} sum-product low-density parity-check codes with finite rate,'' \emph{Physical Review A—Atomic, Molecular, and Optical Physics}, vol.~88, no.~1, p. 012311, 2013.

\bibitem{kasai2011quantum}
K.~Kasai, M.~Hagiwara, H.~Imai, and K.~Sakaniwa, ``Quantum error correction beyond the bounded distance decoding limit,'' \emph{IEEE Transactions on Information Theory}, vol.~58, no.~2, pp. 1223--1230, 2012.

\bibitem{panteleev2021quantum}
P.~Panteleev and G.~Kalachev, ``Quantum {LDPC} codes with almost linear minimum distance,'' \emph{IEEE Transactions on Information Theory}, vol.~68, no.~1, pp. 213--229, 2022.

\bibitem{lin2024quantum}
H.-K. Lin and L.~P. Pryadko, ``Quantum two-block group algebra codes,'' \emph{Physical Review A}, vol. 109, no.~2, p. 022407, 2024.

\bibitem{bravyi2024high}
S.~Bravyi, A.~W. Cross, J.~M. Gambetta, D.~Maslov, P.~Rall, and T.~J. Yoder, ``High-threshold and low-overhead fault-tolerant quantum memory,'' \emph{Nature}, vol. 627, no. 8005, pp. 778--782, 2024.

\bibitem{nouroziapm}
{V. Nourozi, D. G. M. Mitchell, and T. Koike-Akino}, ``{Base-Preserving APM/Voltage Lifts of Bivariate Bicycle Quantum {LDPC} Codes},'' in \emph{IEEE International Conference on Quantum Computing and Engineering (QCE)}, Sep. 2026.

\bibitem{nouroziq}
{V. Nourozi, D. G. M. Mitchell, and T. Koike-Akino}, ``{Q-Learning Base Search Voltage-Labeled Covers for Weight-Six Bivariate-Bicycle Quantum {LDPC} Codes},'' in \emph{Proc. IEEE International Conference on Quantum Computing and Engineering (QCE)}, Sep. 2026.

\bibitem{nourozilp}
V.~Nourozi, M.~Moradi, R.~Smarandache, and D.~G.~M. Mitchell, ``Lifted-product {QLDPC} codes in the polynomial domain,'' \emph{arXiv preprint arXiv:2609.01305}, 2026.

\bibitem{Nourozi2026sep4}
{V. Nourozi, D. G. M. Mitchell, and T. Koike-Akino}, ``{Net-Voltage Screening of Bivariate-Bicycle {T}anner-Graph Covers},'' in \emph{Proc. IEEE International Conference on Quantum Computing and Engineering (QCE)}, Sep. 2026.

\bibitem{smarandache_rank}
R.~Smarandache, D.~G. Mitchell, and A.~G{\'o}mez-Fonseca, ``Structural analysis of generalized quasi-cyclic {LDPC} codes: Rank, design, and generator matrices,'' \emph{IEEE Transactions on Information Theory}, vol.~72, no.~1, pp. 253–267, 2026.

\bibitem{sv12}
R.~Smarandache and P.~O. Vontobel, ``Quasi-cyclic {LDPC} codes: Influence of proto-and {T}anner-graph structure on minimum hamming distance upper bounds,'' \emph{IEEE Transactions on Information Theory}, vol.~58, no.~2, pp. 585--607, 2012.

\end{thebibliography}
\end{document}